\documentclass[journal]{article}                                  

\usepackage{arxiv}
\usepackage{hyperref}
\usepackage{graphicx}          
\usepackage{amsmath} 
\usepackage{amssymb}  

\usepackage{amsthm}
\usepackage{color,soul}
\usepackage{flushend}
\usepackage{cite}
\usepackage{mfirstuc}
\usepackage{mathtools}
\usepackage{xcolor}
\usepackage{array,multirow}
\usepackage{algorithm} 
\usepackage{algorithmic}  
\usepackage[linesnumbered,algo2e,ruled,vlined,norelsize]{algorithm2e} 

\allowdisplaybreaks
\newtheorem{thmm}{Theorem}

\newtheorem{remm}{Remark}
\newtheorem{assm}{Assumption}

\newtheorem{FF}{Residual Feature}

\usepackage{tikz}
\usepackage{textcomp}
\usepackage{hyperref}
\usepackage{lipsum}

\usepackage{authblk}
\title{\LARGE \bf
Generating Sustainability-Targeting Attacks For Cyber-Physical Systems
}

\author[1]{Faysal Ahamed}
\author[1]{Tanushree Roy}
\affil[1]{Department of  Mechanical Engineering, Texas Tech University, Lubbock, TX 79409, US. Emails:~{\tt\small fahamed@ttu.edu, tanushree.roy@ttu.edu}.}

\begin{document}


\maketitle
\pagestyle{empty}

\begin{abstract}
Sustainability-targeting attacks (STA) are a growing threat to cyber-physical system (CPS)-based infrastructure, as sustainability goals become an integral part of CPS objectives. STA can be especially disruptive if it impacts the long-term sustainability cost of CPS, while its performance goals remain within acceptable parameters.
Thus, in this work, we propose a general mathematical framework for modeling such stealthy STA and derive the feasibility conditions for generating a minimum-effort maximum-impact STA on a linear CPS using a max-min formulation. A gradient ascent descent algorithm is used to construct this attack policy with an added constraint on stealthiness. An illustrative example has been simulated to demonstrate the impact of the generated attack on the sustainability cost of the CPS.
\end{abstract}

\section{INTRODUCTION}

As CPS-based critical infrastructure systems increasingly consume an enormous amount of energy, heavily interact with the environment, and focus on sustainability goals and regulatory restrictions, cyberattacks on these systems can directly impact the overall sustainability objectives. In particular, an adversary can specifically target these sustainability goals instead of disrupting system operations solely using conventional cyberattack methodologies such as denial-of-service (DoS), replay attack, false-data injection, and so on. For instance, in 2025, a cyberattack caused undetected water loss for four hours from a dam in Bremanger, Norway.  These emerging threats of sustainability-targeting attacks (STA) motivate our interest in developing a general \textit{offensive security} framework for analyzing their potential risks using an attacker-defender model.

Game-theoretic approaches have been pivotal in modeling adversarial interactions in dynamic systems. \cite{bacsar1998dynamic} developed a foundational framework for dynamic noncooperative games, establishing the basis for minmax decision-making in dynamic environments. These ideas were then extended in \cite{bacsar2008h} to $H_{\infty}$ control and minmax design problems for robust dynamic systems. In the context of cyber-physical systems, \cite{zhu2015game} introduced a game-theoretic framework that captures the interaction between cyber and physical layers in analyzing system robustness and security. However, over the past decade, research has mainly focused on how cyberattacks affect real-time CPS operation; however, attention has recently expanded to STA in systems such as intelligent buildings, smart grids, data centers, etc.
\cite{shekari2021mamiot,li2022reinforcement} showed how energy demand can be manipulated by attacking via Internet of Things (IoT) devices. In \cite{kuaban2023modelling}, the authors explored the energy depletion of IoT device batteries by using diffusion approximation and examined the impacts of ghost energy depletion attacks on device lifetime. Meanwhile, for smart grids, the worst-case power injection attacks are analyzed in \cite{lindstrom2021power} to identify the system vulnerabilities for resilient grid design optimization and reconfiguration.  Similarly, \cite{abraham2024consequence} used dynamic co-simulator models to assess how cyberattacks on smart grids can cause economic losses and threaten the sustainable operation by undermining the critical infrastructure of power systems. Researchers have also evaluated how SYN (synchronize) flood attack and network congestion attack strategies on smart electrical meters can alter or impede energy reporting, leading to poor energy management \cite{kumar2023experimental}. Additionally, \cite{islam2018ohm} explored a novel side-channel attack strategy in multi-tenant data centers to eavesdrop on power consumption data to execute attacks.

Yet, a general framework for formulating, analyzing, and generating these cyberattacks that target the sustainability goals of a CPS to achieve maximum-impact with a minimum effort on the part of an adversary has remained unexplored. Furthermore, while the impacts on the sustainability goals are evaluated across a longer time horizon, the adversary needs to remain stealthy by inducing only negligible effects on real-time system operation. To address these research gaps, the primary focus of this paper is to analyze how the control input can be corrupted to potentially impact the long-term sustainability goals of a CPS. However, the particular mechanism of corruption, such as DoS, false-data injection, etc., remains outside the scope of interest. Thus, our particular contributions in this work are:
\begin{enumerate}
    \item formulating an offensive security framework for STA for a linear CPS with sustainability objectives,
    \item deriving necessary conditions for the existence of a minimum-effort maximum-impact STA policy,
    \item implementation of a Gradient Ascent Descent  (GAD) algorithm to construct an $\alpha$-stealthy minimum-effort maximum-impact STA policy. 
\end{enumerate}

The organization of this paper is as follows: Section~\ref{prob} presents the susceptibility-targeting attack model. Section~\ref{feasibility} presents the feasibility analysis and $\alpha$-stealthy policy generation algorithm using GAD. Section~\ref{sim} presents the simulation results for a linear CPS with a quadratic cost for sustainability. We finally summarize our work in Section~\ref{conclu}.

\textbf{Notations:}
 This paper uses \(\mathbb{R}^a\) as a \(a\)-dimensional Euclidean space. The 2-norm (also called the Euclidean norm) of a vector is $\|x\| := (x^T x)^{1/2}$. 
 The scalar product of two $p$-dimensional vectors is defined as $\big(x,y\big) := \sum_{r=1}^p x_r^T y_r$. Let $I \subset \mathbb{R}$ be a compact interval. The space $C(I)$ denotes the set of continuous function on $I$, and $L^2(I)$ denotes the space of square-integrable functions 
on $I$, such that
$\int_I \|f(t)\|^2 \, dt < \infty$. 

\section{Modeling Framework for STA} \label{prob}

A security framework for CPS is divided into the administrator space and the adversarial space. The administrator space for CPS consists of the plant, a central controller, and a closed communication network connecting the plant and the controller \cite{zhu2011hierarchical}. The CPS output is measured and sent via the cloud to a central controller that implements specific control policies to obtain a control input $u$. In a CPS plant with sustainability objectives, the controller is designed to achieve both operational performance and sustainability goals. The latter is often satisfied by the administrator by minimizing a sustainability cost over a time horizon through an optimal control input $u$. In contrast, in the adversarial space, all or partial system information and measurements are used to craft cyberattacks to fulfill the goals of the adversary. Particularly, an adversary aiming to compromise the administrator’s sustainability objective through STA $\delta$ attempts to maximize the sustainability cost $\mathcal{J}$ through the minimal disruption resources needed to corrupt control input $u$. Fig.~\ref{fig:comb} shows the schematic of the different components of the CPS administrator space and the adversarial space generating the STA to corrupt control input $u$ through cloud communication.
\begin{figure}[h!]
    \centering
    \includegraphics[trim = 0mm 0mm 0mm 0mm, clip,  width=\linewidth]{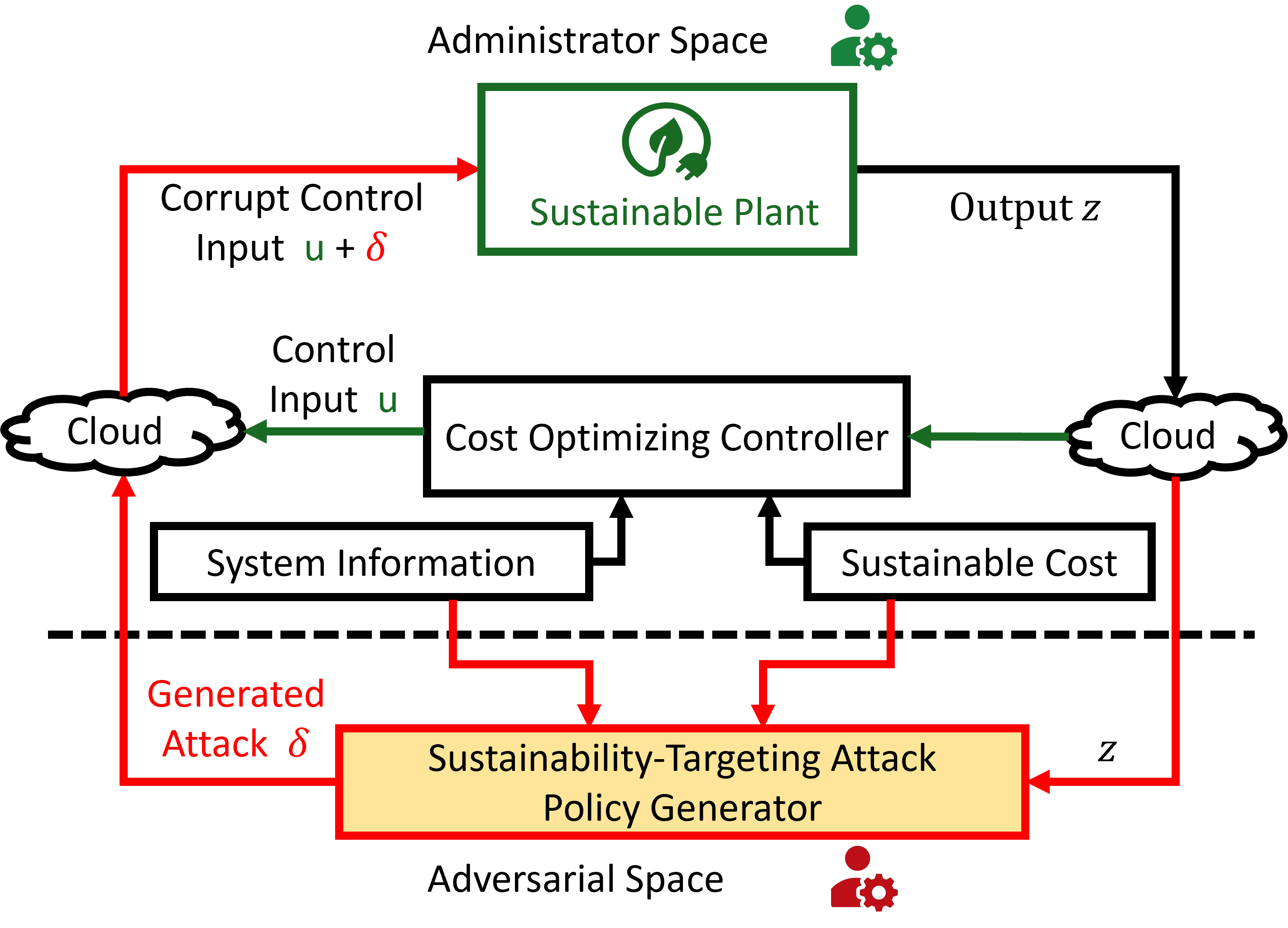}  
    \caption{A schematic of the administrative and adversarial space in the context of sustainability-targeting attacks on CPS.}
    \label{fig:comb}

\end{figure}
To present the mathematical framework for STA, we will first represent a linear CPS that is under an attack that corrupts the control input:
\begin{align}
\dot{x} & = \mathcal{A} x + \mathcal{B}u,  \label{state} \\
u & = -\mathcal{K} z + \delta, \quad z  =  \mathcal{L} x, \label{input}
\end{align}
where \(x \in \mathbb{R}^n\) is a state vector, $u \in \mathcal{N}_{u} \in \mathbb{R}^k$ is the control input, and \(z \in \mathbb{R}^j\) denotes the measured output. The attack generated by the adversary is denoted by $\delta$ such that it belongs to the set of all possible attacks $\mathcal{N}_{\delta}\subset \mathbb{R}^k$, and the control input \(u\) belongs to the admissible set $\mathcal{N}_{u}$. In addition, $\mathcal{A} \in \mathbb{R}^{n \times n}, \mathcal{B} \in \mathbb{R}^{n \times k}$, $ \mathcal{L} \in \mathbb{R}^{j \times n}$, and $\mathcal{K}\in \mathbb{R}^{k \times j}$ are the state, input, observation, and feedback gain matrices, respectively.
Now, let the (positive) sustainability cost function over the time interval $I = [0, T]$ be given by
\begin{align}
    \mathcal{S}(u, x) := \int_{I}w(t, x, u)dt + \theta(x(T)),  \label{cost1}
\end{align}
where \(w(t, x(t))\) is the running cost and \(\theta(x(t))\) is the terminal cost for the system \eqref{state}-\eqref{input}. For instance, the sustainability cost function can capture the loss in building energy utilization \cite{bakar2015energy} or the marginal emission factor for electric vehicles \cite{tu2020electric}.  The goal of the defender here is to minimize this sustainability cost \eqref{cost1} of the CPS (e.g., minimize energy loss in buildings or emissions in electric vehicles) with the least control action $u$ \eqref{input}. The defender or administrator can also minimize their control effort using the cost $ \mathcal{S}(u, x)$. On the other hand, the attacker attempts to maximize the sustainability cost \eqref{cost1} by designing $\delta$, while simultaneously minimizing their (positive) effort cost function  $\mathcal{E}(\delta) := \int_{I} E(\delta) dt$.
Thus, the defender seeks an optimal control policy $u_0$ that minimizes $\mathcal{S}$, while the attacker simultaneously attempts to design a minimum-effort maximum-impact STA policy $\delta_0$ that maximizes total impact-effort cost
\begin{align}
    \mathcal{J}(u,\delta) = \mathcal{S}(u, x)-\mathcal{E}(\delta). \label{cost}
\end{align}
The effort function can be weighted with a regularization parameter $\gamma$ to scale the importance of the adversarial resource cost in generating the attack policy.  In this offensive security framework of dynamic games, the attacker can generate the minimum-effort maximum-impact STA policy by solving the following max-min problem, subject to the dynamics of the CPS \eqref{state}-\eqref{input} 
$\max_{\delta \in \mathcal{N}_{\delta}} \min_{u \in \mathcal{N}_{u}} \mathcal{J}(u, \delta).$
While the sustainability cost function of the system is minimized with respect to the control input \(u\), we can reformulate this optimization in terms of the feedback gain \(\mathcal{K}\) by substituting the control input \eqref{input} to obtain
\begin{equation}\label{maxmin}
\max_{\delta \in \mathcal{N}_{\delta}} \min_{\mathcal{K} \in \mathcal{N}_{\mathcal{K}}} \mathcal{J}(\mathcal{K}, \delta),
\end{equation}
where the feedback gain \(\mathcal{K}\) belongs to the admissible set of the feedback gain values of $ \mathcal{N}_{\mathcal{K}}\subset \mathbb{R}^{k \times j}$. We note that the negative sign for $\mathcal{E}$ in \eqref{cost} allows adversarial effort minimization and the positive sign for $\mathcal{S}$ in \eqref{cost} allows adversarial impact maximization.

\section{Feasibility Analysis And Policy Generation} \label{feasibility}

In this section, we will derive the necessary conditions for the existence of a solution of the max-min problem \eqref{maxmin} under the constraint of the CPS dynamics \eqref{state}-\eqref{input} and prove the existence of a minimum-effort maximum-impact STA policy. This feasibility analysis will also enable us to propose an algorithm for generating the STA under stealthiness constraints using a GAD strategy.

\subsection{Feasibility of the minimum-effort maximum-impact STA}

To derive the necessary conditions for the existence of a minimum-effort maximum-impact STA for a linear CPS, we consider the following assumptions on the access to system information by the attacker and the regularity of the cost functions and the convexity of admissible sets for $u$ and $\delta$.

\begin{assm}[Attacker knowledge] \label{know}

Attackers have complete information of the CPS, i.e., knowledge of \(\mathcal{A}, \mathcal{B}, \mathcal{L}\), the sustainability cost \(\mathcal{S}\) to obtain a specific sustainability goal, and the cost of its efforts $\mathcal{E}(\delta)$. The attacker has full access to the actuator and can modify \(u\) to \(u+\delta\). 
\end{assm}
\begin{remm}
    While complete knowledge of the system can be unrealistic in a practical setting, analyzing STA policy under this assumption is critical. Several attack methods, such as stealthy attack, byzantine attack, white-box attacks, or man-in-the-middle attacks, are often analyzed under the assumption of full-system knowledge \cite{guo2018worst} and show the fundamental threats in the system design. 
\end{remm}
\begin{assm}[Continuous control input \(u\) and attack \(\delta\)] \label{input as} 
The control input and attack signal are continuous and bounded on $I=[0,T]$, i.e, $u,\delta \in C(I;\mathbb{R}^k)$. Consequently, since $I$ is finite, it follows that $u, \delta \in L^2(I, \mathbb{R}^k)$.
  \end{assm}
 \begin{assm}[Convexity of feasible attack set] \label{attack as}
         The set of all admissible attacks $\mathcal{N}_{\delta}\subset \mathbb{R}^k$ is a nonempty, closed, and convex set. This implies that for any $\epsilon \in [0,1]$ and any two attacks $\delta_1, \delta_2 \in \mathcal{N}_{\delta}$, the convex combination $\delta_{\epsilon} = \epsilon \delta_1 + (1- \epsilon) \delta_2 \in \mathcal{N}_{\delta}$.
   \end{assm}
 
\begin{assm}[Regularity of cost function and trajectory] \label{cost as}
    The terminal cost $\theta(x)$ and the running cost $w(t,x)$ are continuous and are continuously differentiable with respect to $x$, and the effort cost $\mathcal E (\delta)$ is continuously differentiable on $\mathcal N_{\delta}$. 
\end{assm}

We derive the closed-loop system of \eqref{state}-\eqref{input} by substituting \(z\) to get \(u= -\mathcal{K}\mathcal{L}x+\delta\), which changes \eqref{state} to:
\begin{align}
  \dot{x} = \mathcal{A} x - \mathcal{B}\mathcal{K}  \mathcal{L} x +  \mathcal{B} \delta.  \label{new state}
\end{align}
We will apply a variational approach to determine the minimum-effort maximum-impact STA \(\delta_0\) in the presence of an optimal control gain \(\mathcal{K}_0\). Thus, we define a Hamiltonian function using \eqref{cost} and \eqref{new state} 
\begin{align}
 \mathcal{H}(t, x, \Omega, \delta)\! :=& (\mathcal{A} x\! -\! \mathcal{B}\mathcal{K}_0  \mathcal{L} x\!+\! \mathcal{B}\delta, \Omega) \!+\! w(t, x). \label{Hamiltonian}
\end{align} 
Our goal here is to show that the necessary conditions for the existence of a minimum-effort maximum-impact STA are equivalent to maximizing the Hamiltonian $\mathcal{H}$. The costate variable \(\Omega\) in \eqref{Hamiltonian} captures the cost of violating the system dynamics during optimization. 
Inspired by \cite{ahmed2011optimal}, we prove our main theorem as presented below.

\begin{thmm}[Necessary conditions for minimum-effort maximum-impact STA]
 Let us consider the closed-loop linear CPS system \eqref{new state} with impact-effort cost \eqref{cost} such that Assumptions ~\ref{know}- ~\ref{cost as} are valid. 
 Under these assumptions, for an optimal controller gain $\mathcal{K}_0$ that minimizes the cost, if a minimum-effort maximum-impact STA $\delta_0$ exists in the sense of the max-min problem \eqref{maxmin} and produces a state trajectory $x_0$ and a costate trajectory $\Omega_0$, then the following three conditions are satisfied for $\forall t\in I$:

\begin{enumerate}
\abovedisplayskip=-\baselineskip
\belowdisplayskip=0pt
\abovedisplayshortskip=-\baselineskip
\belowdisplayshortskip=0pt
    \item \begin{flalign}\label{H_inequality} 
    & \mathcal{H}(t,x_0(t), \Omega_0(t), \delta_0(t)) \! \geqslant \!\mathcal{H}(t,x_0(t), \Omega_0(t), \delta(t)),&& 
\end{flalign}   
\item \begin{flalign}\label{xoo}
     &\dot{x}_0 = \mathcal{H}_\Omega(t,x_0, \Omega_0, \delta_0), \quad x_0(0) = \upsilon,&&
\end{flalign}  
\item \begin{flalign}\label{-H}
     & \dot{\Omega}_0 = -\mathcal{H}_x(t,x_0, \Omega_0, \delta_0),\, \Omega_0(T) = \theta_x(x_0(T)).&& 
\end{flalign}  
\end{enumerate}
\end{thmm}

\begin{proof}
Let us consider an optimal feedback control gain $\mathcal{K}_0\in\mathcal{N}_{\mathcal{K}}$ and the corresponding impact-effort cost
\begin{align}
    \mathcal{J}(\mathcal{K}_0, \delta_0) := \!\int_{I} (w(t, x_0(t))-E(\delta_0))dt+ \theta(x_0(T)), \label{Costfunction}
\end{align}
where the state $x_0(t):=x(t,\mathcal{K}_0, \delta_0)$ is the solution of 
\begin{align}
  \dot{x}_0 = \mathcal{A} x_0 - \mathcal{B}\mathcal{K}_0  \mathcal{L} x_0 +  \mathcal{B}\delta_0,\ x_0(0) = \upsilon,\, t \in I. \label{state12}
\end{align}
Now, we suppose that $\delta_0 \in \mathcal{N}_{\delta}$ is the minimum-effort maximum-impact attack for the optimal gain $\mathcal{K}_0$ and $\delta \in \mathcal{N}_{\delta}$ is any other attack. Since \(\mathcal{N}_{\delta}\) is a closed convex set, for any $\epsilon \in (0,1)$ the combination of $\delta$ and $\delta_0$ such as $\delta_{\epsilon} := \delta_0 + \epsilon (\delta-\delta_0)$ is also in the admissible set $\mathcal{N}_{\delta}$. This perturbed attack $\delta_{\epsilon}$ helps us examine how the impact-effort cost changes when we change the STA  $\delta_0$ subtly.
 Now, the solution of \eqref{state12} corresponding to the perturbed attack \(\delta_\epsilon\) is denoted as \(x_{\epsilon}\). Since \(\delta_0\) is the maximum-impact minimum effort attack, we have \(\mathcal{J}(\mathcal{K}_0, \delta_0) \geqslant \mathcal{J}(\mathcal{K}_0, \delta_\epsilon)\), \(\forall \epsilon \in (0,1)\). Next, dividing the inequality by $\epsilon$ and using \eqref{cost} we obtain: 
    $\frac{1}{\epsilon} \left( \mathcal{J}(\mathcal{K}_0, \delta_\epsilon) - \mathcal{J}(\mathcal{K}_0, \delta_0) \right) 
    = \frac{1}{\epsilon} [ \int_{I} (w(t, x_\epsilon(t)) - w(t, x_0(t))) \, dt -\int_{I}(E(\delta_{\epsilon})-E(\delta_0))dt + \theta(x_\epsilon(T)) - \theta(x_0(T))]$ $\leqslant 0.$ 
As $\epsilon \to 0$, this yields the Gateaux derivative $d\mathcal{J}(\mathcal{K}_0, \delta_0)$ of the impact effort cost $\mathcal{J}$ at $\delta_0$, which implies that $d\mathcal{J}(\mathcal{K}_0, \delta_0)\leqslant 0$. Additionally, \eqref{cost} shows that this Gateaux derivative can be written in terms of two derivatives: $d\mathcal{J}(\mathcal{K}_0, \delta_0)=d\mathcal{S}(\mathcal{K}_0, \delta_0)-d\mathcal{\mathcal{E}}(\mathcal{K}_0, \delta_0)$, where
\begin{align}\nonumber
    d\mathcal{S}(\mathcal{K}_0, \delta_0)=& \lim_{\epsilon \rightarrow 0} \frac{1}{\epsilon} \Big[ \int_{I} (w(t, x_\epsilon(t)) - w(t, x_0(t))) \, dt\\ 
    &\hspace{2mm}+ \theta(x_\epsilon(T)) - \theta(x_0(T))\Big], \label{cals}\\
   d\mathcal{\mathcal{E}}(\mathcal{K}_0, \delta_0) = &\lim_{\epsilon \rightarrow 0}  \frac{1}{\epsilon}\int_{I}(E(\delta_{\epsilon})-E(\delta_0))dt.
\end{align}
Since $\delta_0$ is the minimum-effort attack, $E(\delta_0)\leqslant E(\delta_\epsilon)$, i.e. $ d\mathcal{\mathcal{E}}(\mathcal{K}_0, \delta_0)\geqslant 0$. Therefore, a sufficient condition for satisfying $d\mathcal{J}(\mathcal{K}_0, \delta_0)$ is $d\mathcal{S}(\mathcal{K}_0, \delta_0)\leqslant 0$. Now, \eqref{cals} yields
\begin{align}
&d\mathcal{S}(\mathcal{K}_0, \delta_0) = \mathbb{K}_y[\mathcal{S}_x] \leqslant 0, \label{Gateaux}\\
&\text{where operator } \mathbb{K}_y[\mathcal{S}_x]:= (\mathcal{S}_x,y),\label{KJ}
\end{align}
  and function $y(t) := \lim_{\epsilon \to 0} \tfrac{x_\epsilon(t) -  x_0(t)}{\epsilon}$.
Using Assumption \ref{cost as}, $x, x_{\epsilon} \in C(I)$ and in turn $y(t) \in C(I)$ is the solution of 
\begin{align}
\dot{y} = \mathcal{A} y - (\mathcal{B} \mathcal{K}_0  \mathcal{L} ) y + \mathcal{B}(\delta- \delta_0), \quad y(0) = 0. \label{variational}
\end{align}
This variational equation \eqref{variational} shows how the perturbation $y(t)$ in the state evolves with a small variation in the attack $\delta_0$. Since the initial conditions $x_0(0)=x_\epsilon(0)=\nu$ are fixed with respect to variations in the attack term, the state variation is $y(0)=0$. By Assumption~\ref{input as}, attacks $\delta,\delta_0 \in L^2(I;\mathbb{R}^k)$ are bounded functions; thus the driving term in \eqref{variational}, $\mathcal{B}(\delta  - \delta_0)\in L^2(I)$.
This guarantees the existence and uniqueness of the solution $y \in C(I)$.  
Additionally, this implies that there exists a continuous map
   $\mathcal{B} (\delta  - \delta_0) \to y. $
  Next, from Assumption~\ref{cost as}, the terms $\int_I \big(w_x(t,  x_0(t) \big)dt$ and $\theta_x \big( x_0(T)\big)$ are continuously differentiable with respect to $x$. Therefore, the map $y \to \mathbb{K}_y[\mathcal{S}_x]$ is also continuous.
So, the composition of the two maps also yields a continuous linear mapping 
\begin{align}
\mathcal{B}(\delta  - \delta_0) \to \mathbb{K}_y[\mathcal{S}_x]. \label{newstate}
\end{align}
Thus, \eqref{Gateaux} and \eqref{newstate} imply that the Gateaux derivative $d\mathcal{S}$ is a continuous linear functional of $\mathcal{B}(\delta  - \delta_0)\in L^2(I)$. Therefore, we can now apply the Riesz representation theorem \cite{rudin1991functional} that 
guarantees the existence of a unique function $\Omega_0 \in L^2(I;\mathbb{R}^n)$ such that the Gateaux derivative of the cost function can be expressed as an inner product with $\Omega_0$. Additionally, using \eqref{Gateaux} we can infer $\forall \delta \in \mathcal{N}_{\delta}$
\begin{align}
 d \mathcal{S} 
= \int_I  \Big(\mathcal{B}(\delta  - \delta_0), \Omega_0 \Big) \, dt\leqslant 0 
\label{gd}
\end{align}
This $\Omega_0$ can in fact be defined as our costate from the Hamiltonian \eqref{Hamiltonian}. Consequently, observing \eqref{Hamiltonian}, we can add more terms to the inequality $\int_I \big(\mathcal{B} \delta , \Omega_0 \big) dt  \leqslant  \int_I \big(\mathcal{B} \delta_0, \Omega_0 \big) dt$  to obtain $\int_I [\mathcal{H}( t,x_0, \Omega_0, \delta) - \mathcal{H}(t,x_0, \Omega_0,\delta_0)] dt \leqslant 0.$ Since this inequality is true $\forall \delta\in \mathcal{N}_{\delta}$, for some $s\in I$ we define $\delta(t)= \delta_0(t), \forall t\in I\backslash (s-\omega,s+\omega)$ and $\delta(t)=\mathcal{D}\in\mathcal{N}_{\delta}$ otherwise. This spike variation choice of $\delta$, reduces the integral to $\int_{s-\omega}^{s+\omega}[\mathcal{H}( t,x_0, \Omega_0, \mathcal{D}) - \mathcal{H}(t,x_0, \Omega_0,\delta_0)] dt \leqslant 0$ \cite{ahmed2006dynamic}. We note here that $\mathcal{H}$ is an $L^2(I)$ on a measurable set $(s-\omega,s+\omega)$. Hence, we can use the Lebesgue differential theorem \cite{weiss2015course} by dividing both side by $\omega$ and set $\omega\to 0$, to obtain the point-wise inequality   $\mathcal{H}( t,x_0, \Omega_0, \delta) - \mathcal{H}(t,x_0, \Omega_0,\delta_0)\leqslant 0,$ for almost all $ t\in I$ and any $\delta\in\mathcal{N}_{\delta}$, proving condition \eqref{H_inequality}.

Next, to prove the condition \eqref{xoo},  we differentiate the Hamiltonian \eqref{Hamiltonian} with respect to the costate variable \(\Omega\). This produces the right-hand side of the system dynamics expressed in \eqref{state12}, and we recover condition \eqref{xoo}

Subsequently, to prove the condition \eqref{-H}, we will first derive the costate dynamics. To achieve this, we will rearrange the terms in \eqref{variational} to isolate $\mathcal{B}(\delta- \delta_0)$ and substitute into \eqref{gd} 
\begin{align}
&d\mathcal{S}= \int_I  \Big( \dot{y} - [\mathcal{A} y - (\mathcal{B} \mathcal{K}_0 \mathcal{L}) y], \Omega_0 \Big) \, dt. \label{C12}
\end{align}
We apply integration by parts to the integral in \eqref{C12} with $y(0)=0$ to obtain
$\int_I  \big( \dot{y} - [ \mathcal{A} y - (\mathcal{B} \mathcal{K}_0 \mathcal{L}) y], \Omega_0 \big)  \, dt
= \big(y(T), \Omega_0(T))  - \int_I  \big( y, \dot{ \Omega}_0 + \mathcal{A}^T\Omega_0 - (\mathcal{B}\mathcal{K}_0\mathcal{L})^T \Omega_0 \big)  \, dt.$
Replacing this in \eqref{C12}, then using \eqref{Gateaux}-\eqref{KJ} and comparing, we obtain the costate dynamics:
\begin{align}
\dot{ \Omega}_0   =  -\left[\mathcal{A}^T\Omega_0 -(\mathcal{B}\mathcal{K}_0\mathcal{L})^T \Omega_0+ w_x (t,  x_0(t))\right] , \label{costates}
\end{align}
with the terminal condition 
   $ \Omega_0(T) = \theta_x ( x_0(T)).$ 
Now, if we differentiate the Hamiltonian \eqref{Hamiltonian} with respect to the state variable $x$,  we obtain the right-hand side of \eqref{costates}. This yields the third condition specified in \eqref{-H} and completes the proof.
\end{proof}
\begin{thmm}[Existence of minimum-effort maximum-impact STA]
    Consider the system \eqref{new state} with the sustainability impact-effort cost \eqref{cost}, a functional of the attack $\delta\in \mathcal{N}_{\delta}$ and for any optimal control law $\mathcal{K}_0 \in \mathcal{N}_{\mathcal{K}}$. If Assumptions~\ref{input as}-\ref{cost as} hold, then there exists at least one minimum-effort maximum-impact STA $\delta_0\in \mathcal{N}_{\delta}$ that achieves a maximum value of an impact-effort cost. 
\end{thmm}
\begin{proof}
   Using Assumption~\ref{input as} and invoking Alaoglu's Theorem \cite{rudin1991functional}, we can infer that $\mathcal{N}_{\delta}$ is a weak-star ($w^*$) compact set. This implies that to prove the existence of the maximum of the cost \eqref{cost}, it is sufficient to prove that $\delta \rightarrow J(\mathcal{K},\delta)$ is sequentially weak-star continuous. 

   Since $\mathcal{N}_{\delta}$ is a weak-star compact set,  let us consider a sequence of $\delta_i\in \mathcal{N}_{\delta}, \forall i\in \mathbb{N}$ such that $\delta_i\xrightarrow[]{w^*}\delta_0 \in\mathcal{N}_{\delta}$ (by compactness of $\mathcal{N}_{\delta}$). Let $x_i$ and $x_0$ denote the solutions to \eqref{new state} for $\delta_i$ and $\delta_0$, respectively. Then, using \eqref{new state}, Assumptions~\ref{input as}-\ref{cost as}, and Gronwall inequality, we can readily derive 
    \begin{align}\label{thm2_eq}
        \|x_0(t)-x_i(t)\|\leqslant \|F_i(t)\| + q\int_0^t\!\!\! \|F_i(s)\|\,d s,
    \end{align}
    where some constant $q$ dependent on $A,B,\mathcal{K}_0\mathcal{L}$ and 
        $F_i(t) =\int_0^t B(\delta_0(s)-\delta_i(s))\, ds, \forall t\in I.$ 
Now for any $\phi\in \mathbb{R}^n$, we can write $(F_i(t),\phi) \rightarrow 0$ as $i\rightarrow \infty, \forall t$, since $\delta_i\xrightarrow[]{w^*}\delta_0 $. However, by Assumptions~\ref{input as}-\ref{attack as}, $F_i(t)$ is bounded and is in an finite-dimensional space $\mathbb{R}^n$. Hence, weak and strong convergence are the same, i.e. $\|F_i(s)\|\rightarrow 0$ as $i\rightarrow \infty, \forall t$. Next, by the Lebesgue dominated convergence theorem, the second term in \eqref{thm2_eq} also converges to $0$, which in turn implies $\lim_{i\rightarrow \infty}x_i(t)=x_0(t), \forall t\in I$. Lastly, using the continuity Assumption~\ref{cost as}, we prove the $w^*$ continuity of $\lim_{i\rightarrow \infty} J(\mathcal{K},\delta_i)=J(\mathcal{K},\delta_0)$ on the set $\mathcal{N}_{\delta}$. 
\end{proof}

With the derivation of the necessary conditions for the existence of the minimum-effort maximum-impact STA $\delta_0$ in Theorems 1 and 2, we will present an algorithmic approach to construct it under real-time stealthy criteria.

\subsection{Gradient ascent-descent (GAD) algorithm for stealthy STA generation}

In this subsection, we describe a GAD algorithm \cite{wang2023optimal} that is used to obtain the minimum-effort maximum-impact  STA policy for our linear dynamical system \eqref{state}-\eqref{input} under an $\alpha$-stealthy constraint. Our goal is to \textit{algorithmically obtain a STA that will maximize its impact on the sustainability goals of the CPS in the long-term with minimum effort, while remaining $\alpha$-stealthy in real-time}.  
 
In GAD, the gradient descent method is followed by the gradient ascent to calculate the optimal feedback gain \(\mathcal{K}_0\) and the minimum-effort maximum-impact  STA \(\delta_0\), respectively. 
During the gradient descent step, we update the feedback control gain $\mathcal{K}$ by using the gradient descent equation 
\begin{equation}
\mathcal{K}^{l+1} =\mathcal{K}^{l} - \lambda_{\mathcal{K}} \nabla_{\mathcal{K}}\mathcal{H}, \label{gradient_K}
\end{equation}
to find the optimal value of feedback control gain \(\mathcal{K}\) for which the cost function is minimum. Similarly, in the gradient ascent step, the value of the attack \(\delta\) is updated progressively to find the desired minimum-effort maximum-impact  STA by maximizing the impact-effort cost function using gradient-ascent step
\begin{equation}
    \delta^{l+1}=\delta^{l} + \lambda_{\delta} \nabla_{\delta}\mathcal{H}.  \label{gradient d}
\end{equation}
Here, $\lambda_{\mathcal{K}}$, $\lambda_{\delta}>0$ are respectively the learning rates for the ascent and descent parts. $\mathcal{K}^l$ and $\delta^l$ are the gain and attack in the $l^{th}$ iteration. Now if $|\mathcal{J}^{l+1} - \mathcal{J}^{l}| < \eta$, the GAD algorithm is stopped and  $\mathcal{K}^{l+1}$  is chosen as the the optimal gain $\mathcal{K}_0$ and $\delta^{l+1}$ provides an optimal STA $\tilde{\delta}_0$. This attack choice is then designed to ensure $\alpha$-stealthiness. 

For detection, a defender often utilizes a residual generator that monitors the real-time operation of the CPS by measuring the system output $z$, identifying any deviation from the desired system output $z_{ref}$. The residual $r=\|z_{ref}-z\|_2$ is then compared with a pre-defined threshold $\alpha$ \cite{Ding,Ghosh2023SecuritySwitching}. If $r=\|z_{ref}-z\|_2\geqslant \alpha$ then a decision of attack is made, and a no-attack decision otherwise. Thus, any attack $\delta$ is considered to be $\alpha$-stealthy if it evades detection by generating an output $z(\mathcal{K},\delta)$ such that $r=\|z_{ref}-z(\mathcal{K},\delta)\|_2\leqslant \alpha$ \cite{teixeira2012attack}. 

To enforce stealthiness, 
the final stealthy STA $\delta_0$ is obtained through an optimal scaling 
\begin{align}
    \delta_0(t)= \mu_0(t) \tilde\delta_0(t), \forall t \in I,   \label{muo}
\end{align}
where $\mu^*(t)\in (0,1]$ is selected as the pointwise maximum admissible scaling factor using
\begin{align}\label{mu_star}
    \!\!\!\mu_0(t)\!:=\!\max\!\left\{\mu \in (0,1]:\!\|z_{ref}\!-\!z(\mathcal{K}_0,\mu{\tilde\delta_0(t)})\|_2\leqslant \alpha\right\}, \forall t.
\end{align}
The algorithm begins with the initial states of the system and initial guesses for the feedback gain and attack from their respective feasible sets.
Now, the adjoint state $\Omega$ is computed by solving the costate equation \eqref{-H} over the whole time window. Then we will evaluate $\nabla_{\mathcal{K}}\mathcal{H}$ using \eqref{Hamiltonian}  and obtain the updated feedback gain from \eqref{gradient_K}. With this updated control gain and costate variable, we compute  $\nabla_{\delta}\mathcal{H}^l $  from \eqref{Hamiltonian} and obtain the updated STA policy  $\tilde\delta_0(t)$ from \eqref{gradient d}. Then $\tilde \delta_0(t)$ is post-processed via a scaling operation using \eqref{muo}-\eqref{mu_star} to ensure that the final STA becomes $\alpha$-stealthy. 
The stopping criterion for this algorithm is set until the change in cost function \eqref{Costfunction} between successive iterations is less than a set tolerance $\eta$. 
The details of this GAD-based STA generation method have been presented in Algorithm~\ref{alg:desicion}.

\begin{algorithm2e}[h!]

\caption{Gradient Ascent Descent Algorithm}\label{alg:desicion}
\SetKwFunction{FR}{Init{\_}dynamics}
\SetKwFunction{FIC}{Cost{\_}Fn}
\SetKwFunction{FCO}{Costate{\_}Est}
\SetKwFunction{FGK}{Grad{\_}$\mathcal{K}${\_}Est}
\SetKwFunction{FF}{Grad{\_}Desc}
\SetKwFunction{FGD}{Grad{\_}$\delta${\_}Est}
\SetKwFunction{FQ}{Grad{\_}Asc}
\SetKwFunction{FD}{Comp{\_}dynamics}
\SetKwFunction{FM}{Comp{\_}Cost}
\SetKwFunction{FN}{cost{\_}Fn{\_}check}
\KwIn{Number of GAD iterations $N$, initial guesses for gain $\mathcal{K}^1$ and attack $\delta^1$, learning rates  $\lambda_{\mathcal{K}},  \lambda_{\delta}$, tolerance $\eta$.}

\KwOut{Optimal gain $\mathcal{K}_0$, desired STA $\delta_0$}
    $x^{1}, \mathcal{J}^{1} \, \leftarrow$  \texttt{Cost{\_}Fn($\mathcal{K}^{1}, \delta^{1}$)}\\
\For{$l=1:N$}{
    $\Omega^l,  \mathcal{K}^{l+1}   \leftarrow$\texttt{ Grad{\_}Desc($ x^{l}, \mathcal{K}^l, \delta^l$)}
    \\
    
     $\delta^{l+1}   \leftarrow$ \texttt{Grad{\_}Asc($x^{l},  \mathcal{K}^{l+1}, \delta^l, \Omega^l$)}\\
    
  $x^{l+1}, \mathcal{J}^{l+1}\,\leftarrow$  \texttt{Cost{\_}Fn($\mathcal{K}^{l+1}, \delta^{l+1}$)}\\

    \If{$|\mathcal{J}^{l+1} - \mathcal{J}^{l}| < \eta$}{  
    
    \texttt{ $\mathcal{K}_0 \leftarrow \mathcal{K}^{l+1},\delta_0 \leftarrow \mu_0\delta^{l+1}$ } from \eqref{muo}-\eqref{mu_star}\\
    \textbf{break}}
}
     \SetKwProg{Fn}{function}{:}{\KwRet}
     \Fn{\FIC{$\mathcal{K}^l, \delta^l$}}{
         Compute $x^{l}, \forall t$ from \eqref{state}; $\mathcal{J}^{l}$ from \eqref{Costfunction}.\\
        \KwRet $x^{l}, \mathcal{J}^{l}$ ;
         }
       
           \SetKwProg{Fn}{function}{:}{\KwRet}
    \Fn{\FF{$x^{l}, \mathcal{K}^{l}, \delta^{1}$}}{
       Compute $\Omega^{l}, \forall t$ from \eqref{-H}; $\nabla_{\mathcal{K}}\mathcal{H}^l$ from \eqref{Hamiltonian}; $\mathcal{K}^{l+1}$  from \eqref{gradient_K}.\\
        \KwRet $\Omega^{l}, \mathcal{K}^{l+1}$ \;
        }
        
        \SetKwProg{Fn}{function}{:}{\KwRet}
    \Fn{\FQ{$x^{l},  \mathcal{K}^{l+1}, \delta^l, \Omega^{l}$}}{
     Compute $\nabla_{\delta}\mathcal{H}^l $  from \eqref{Hamiltonian};  $\delta^{l+1}$ from \eqref{gradient d} \\
        \KwRet $\delta^{l+1} $;
        }

\end{algorithm2e}

\section{SIMULATION RESULTS}\label{sim}

We will now present a simulation case study to show the efficacy of Algorithm~\ref{alg:desicion} in generating the minimum-effort maximum-impact $\alpha$-stealthy STA for a linear CPS. We also measure the impact of this desired generated attack $\delta_0$ on the sustainability cost of the system compared to a no-attack scenario \((\delta=0)\) with optimal sustainability performance. 
In this work, the linear system \eqref{state}-\eqref{input} has been chosen as $\mathcal{A}=[1\ 2;1\ 2]$, $\mathcal{B}=[2\ 1]^T$. This is a full-state feedback system, i.e., $y=x$, and our goal here is to reach the zero steady-state.
To capture the impact-effort costs for this illustrative example, we have chosen an illustrative quadratic sustainability impact-effort cost:
\begin{equation}
    \mathcal{J}(u, \delta) = \int_0^T \big(\|x(t)\|^2 + \|u(t)\|^2 - \gamma \|\delta(t)\|^2 \ \big) dt. \label{cost_function}
\end{equation}
where the effort cost function $\mathcal{E}(\delta)= \int_0^T  \gamma \|\delta(t)\|^2 dt$ and effort regularization parameter $\gamma=1.$
 The time horizon is chosen to be $T=100s.$ The corresponding Hamiltonian for \eqref{cost_function} can be defined as
    $\mathcal{H}=  \Omega^T [(A-BK) x(t) + B \delta(t)] + \|x(t)\|^2 
    +  \|u(t)\|^2 - \gamma\|\delta(t)\|^2  $.
The goal of the adversary here is to remain $\alpha$-stealthy where $\alpha=0.003$ i.e. $\|x(\mathcal{K}_0, \delta_0)\|_2\leqslant 0.003.$
 
 \begin{figure}[H]
    \centering
    \includegraphics[trim=0 2mm 8mm 0, clip, width=.6\linewidth]{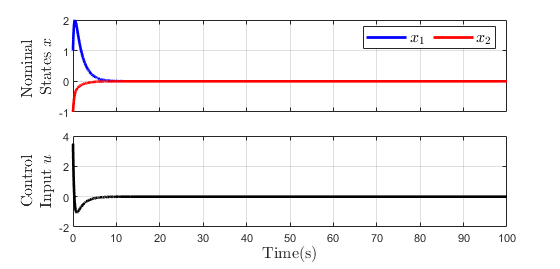}
    \caption{System under nominal conditions: top plot shows controlled system trajectories and bottom plot shows the LQR-based control input.}
    \label{fig:case1}
\end{figure}

For this system, Figure~\ref{fig:case1} shows the performance for a nominal situation, under the Linear Quadratic Regulator (LQR) control gain \cite{zhao2023optimal} $\mathcal{K}=[1.26 \ 4.76]$. The top plot of Figure~\ref{fig:case1} shows that the state trajectories are stabilized under no attack, and the total cost $\mathcal{J} = 8.47$ is minimized. The LQR control input $u(t)$ also remains low throughout the time horizon and reaches $u=0$, as shown in the bottom plot of Figure~\ref{fig:case1}. This confirms that the administrator achieves improved sustainability performance.
\begin{figure}[h!]
    \centering
    \includegraphics[trim=0 1mm 8mm 0, clip, width=.6\linewidth]{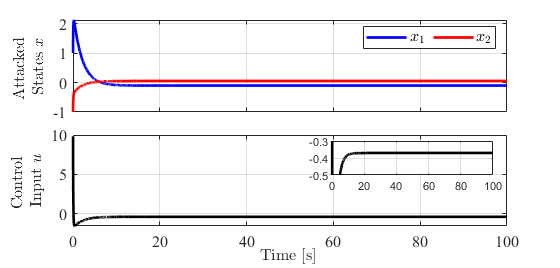}
    \caption{System performance under attack: top plot shows deviations from the reference trajectories and bottom plot shows nominal control input.}
    \label{fig:case2}
\end{figure}

Next, we analyze how this system behaves when the defender tries to minimize the sustainability cost  $\mathcal{S}$ while the attacker simultaneously tries to maximize impact-effort cost, while remaining $\alpha$-stealthy. Using Algorithm~\ref{alg:desicion}, the attacker generates the minimum-effort maximum-impact  STA $\delta_0=0.37$.  On the other hand, the administrator sets the optimal feedback gain as $\mathcal{K}_0=[2.74 \ 12.55]$ to achieve the best-case control under attack, and the corresponding control input reaches $-0.37$, vide inset of the bottom plot of Figure~\ref{fig:case2}. Under the minimum-effort maximum-impact STA, the system trajectories thus settle to $x=[-0.10 \ 0.05]^T$, as shown in the top plot of Figure~\ref{fig:case2}. These trajectories satisfy the $\alpha$-stealthy criteria. The sustainability cost under a minimum-effort maximum impact STA policy is given by $\mathcal{S}= 23.81$, which is a 181\% increase compared to the nominal sustainability cost.

\section{CONCLUSIONS}\label{conclu}

In this work, we proposed an offensive security framework for modeling a cyberattack that can compromise the long-term sustainability goals of an administrator using a min-max optimization. We have derived the necessary conditions for generating a minimum-effort maximum-impact STA policy using the Hamiltonian function and costate dynamics for a linear CPS with sustainability-driven objective. Additionally, we have provided an algorithmic strategy using gradient ascent and descent methods to construct STAs that remain $\alpha$-stealthy over real-time monitoring. The algorithm is implemented for an illustrative example, and the simulation results show that the minimum-effort maximum impact STA significantly increases the sustainability cost of the administrator over a time horizon when compared to nominal operation. This impact is achieved while simultaneously evading detection. Our future work will focus on implementing our proposed framework to generate STA for real-world systems with application-specific sustainability metrics.






\bibliography{ref.bib,Ref_Troy,ref_sus}

\end{document}